\documentclass[conference,a4paper]{IEEEtran}
\usepackage{amsmath,amsthm,amssymb,dsfont}
\usepackage{graphicx,enumerate,url,multirow,multicol}

\usepackage{epsfig}
\usepackage{epstopdf}
\usepackage{tensor}
\usepackage{mystyle}

\newcommand{\reg}{R}
\newcommand{\norma}[1]{\norm{#1}_{\Aa}}

\newcommand{\cA}{C_{\Aa}}
\newcommand{\cPhi}{C_{\Phi}}
\newcommand{\icPhi}{{\cPhi}^{-1}}

\newcommand{\Tx}{{T}}
\newcommand{\ex}{{e}}
\newcommand{\Txper}{{S}}
\newcommand{\Txo}{{T_{0}}}
\newcommand{\exo}{{e_{0}}}
\newcommand{\Txoper}{{S_{0}}}
\newcommand{\st}{\quad\text{s.t.}\quad}
\newcommand{\IC}{\mathrm{IC}}

\newcommand{\GT}{\Gamma}
\newcommand{\AT}{\Xi}
\newcommand{\SC}[1]{{\rm(SC$_{#1}$)}}
\newcommand{\INJ}[1]{{\rm(INJ$_{#1}$)}}

\newcommand{\myparagraph}[1]{\noindent{\bf{#1}~}}

\graphicspath{{./images/}}

\begin{document}
%
\title{Stable Recovery with Analysis Decomposable Priors}

\author{
\begin{tabular}{@{}c@{}@{}cc@{}@{}c@{}}
Jalal M. Fadili & Gabriel Peyr\'e and Samuel Vaiter & Charles-Alban Deledalle & Joseph Salmon \\
GREYC & CEREMADE & IMB & LTCI \\
CNRS-ENSICAEN-Univ. Caen & CNRS-Univ. Paris-Dauphine & CNRS-Univ. Bordeaux 1 & CNRS-T\'el\'ecom ParisTech \\
Caen, France & Paris, France & Bordeaux, France & Paris, France \\
\end{tabular}
}


%


\maketitle

\begin{abstract}
In this paper, we investigate in a unified way the structural properties of solutions to inverse problems. These solutions are regularized by the generic class of semi-norms defined as a decomposable norm composed with a linear operator, the so-called analysis type decomposable prior. This encompasses several well-known analysis-type regularizations such as the discrete total variation (in any dimension), analysis group-Lasso or the nuclear norm. Our main results establish sufficient conditions under which uniqueness and stability to a bounded noise of the regularized solution are guaranteed. Along the way, we also provide a strong sufficient uniqueness result that is of independent interest and goes beyond the case of decomposable norms.
\end{abstract}



%
\IEEEpeerreviewmaketitle

\section{Introduction}
\subsection{Problem statement}
Suppose we observe 
\eq{
y = \Phi x_0 + w, \qwhereq \norm{w}_2 \leq \varepsilon ~,
}
where $\Phi$ is a linear operator from $\RR^N$ to $\RR^M$ that may have a non-trivial kernel.
We want to robustly recover an approximation of $x_0$ by solving the optimization problem
\eql{
\label{eq:lassoA}
x^\star \in \uArgmin{x \in \RR^N} \tfrac{1}{2}\norm{y - \Phi x}_2^2 + \la \reg(x) ~,
}
where
\[
\reg(x) := \norma{L^* x} ~,
\]
with $L: \RR^P \to \RR^N$ a 
linear operator, and $\norma{\cdot}: \RR^P \to \RR^+$ is a decomposable norm in the sense of~\cite{CandesDecomposable12}. Decomposable regularizers are intended to promote solutions conforming to some notion of simplicity/low complexity that complies with that of $u_0=L^* x_0$. This motivates the following definition of these norms. Throughout the paper, given a subspace $V \subset \RR^P$, we will use the shorthand notation $L_V=L\proj_{V}$, $L_V^*=\proj_{V}L^*$, and $\alpha_V=\proj_{V}\alpha$ for any vector $\alpha \in \RR^P$, where $\proj_{V}$ (resp. $\proj_{V^\perp}$) is the orthogonal projector on $V$ (resp. on its orthogonal complement $V^\perp$). 
\begin{defn}\label{def:decomposable-norm}
A norm $\norma{\cdot}$ is \emph{decomposable} at $u \in \RR^P$ if:
\begin{enumerate}[(i)]
\item there is a subspace $T \subset \RR^P$ and a vector $e \in T$ such that
      \eq{
        \partial \norma{\cdot}(u) =
        \enscond{\alpha \in \RR^P}
        {\alpha_{T} = e \quad \text{and} \quad \norma{\alpha_{T^\perp}}^* \leq 1}
      }
\item and for any $z \in T^\perp$, $\norma{z} = \sup_{v \in T^\perp, \norma{v}^* \leq 1} \dotp{v}{z}$, where $\norma{\cdot}^*$ is the dual norm of $\norma{\cdot}$.
\end{enumerate}
\end{defn}
From this definition, it can be easily proved, using Fenchel identity, that $u \in T$ whenever $\norma{\cdot}$ is decomposable at $u$. Popular examples covered by decomposable regularizers are the $\ell_1$-norm, the $\ell_1$-$\ell_2$ group sparsity norm, and the nuclear norm~\cite{CandesDecomposable12}. 

\subsection{Contributions and relation to prior work}
In this paper, we give a strong sufficient condition under which~\eqref{eq:lassoA} admits a unique minimizer. From this, sufficient uniqueness conditions are derived. Then we develop results guaranteeing a stable approximation of $x_0$ from the noisy measurements $y$ by solving~\eqref{eq:lassoA}, with an $\ell_2$-error that comes within a factor of the noise level $\varepsilon$. This goes beyond~\cite{CandesDecomposable12} who considered identifiability under a generalized irrepresentable condition in the noiseless case with $L=\Id$. $\ell_2$-stability for a class of decomposable priors closely related to Definition~\ref{def:decomposable-norm}, is also studied in~\cite{WainwrightDecomposable12} for $L=\Id$ and general sufficiently smooth data fidelity. Their stability results require however stronger assumptions than ours (typically a restricted strong convexity which becomes a type of restricted eigenvalue property for linear regression with quadratic data fidelity). The authors in~\cite{Chandrasekaran12} provide sharp estimates of the number of generic measurements required for exact and $\ell_2$-stable recovery of models from random partial information by solving a constrained form of~\eqref{eq:lassoA} regularized by atomic norms. This is however restricted to the compressed sensing scenario. Our results generalize the stability guarantee of~\cite{Haltmeier12} established when the decomposable norm is $\ell_1$ and $L^*$ is the analysis operator of a frame. A stability result for general sublinear functions $R$ is given in~\cite{GrasmairPositively11}. The stability is however measured in terms of $R$, and $\ell_2$-stability can only be obtained if $R$ is coercive, \ie $L^*$ is injective. 

At this stage, we would like to point out that although we carry out our analysis on the penalized form~\eqref{eq:lassoA}, our results remain valid for the data fidelity constrained version but obviously with different constants in the bounds. We omit these results for obvious space limitations.


\section{Uniqueness}
\label{sec:uniqueness}
\subsection{Main assumptions}
We first note that traditional coercivity and convexity arguments allow to show that the set of (global) minimizers of~\eqref{eq:lassoA} is a non-empty compact set if, and only if, $\ker(\Phi) \cap \ker(L^*) = \{0\}$.

The following assumptions will play a pivotal role in our analysis.\\

\myparagraph{Assumption \SC{x}} There exist $\eta \in \RR^M$ and $\alpha \in \partial\norma{\cdot}(L^*x)$ such that the following so-called source (or range) condition is verified:
\[
\Phi^*\eta=L\alpha \in \partial\reg(x) ~.
\]
\myparagraph{Assumption \INJ{T}} For a subspace $T \subset \RR^{P}$, $\Phi$ is injective on $\ker(L_{T^\perp}^*)$.\\

It is immediate to see that since $\ker(L^*) \subseteq \ker(L_{T^\perp}^*)$, \INJ{T} implies that the set of minimizers is indeed non-empty and compact.

\subsection{Strong Null Space Property}
We shall now give a novel strong sufficient uniqueness condition under which problem~\eqref{eq:lassoA} admits exactly one minimizer. 
\begin{thm}
\label{theo:uniquenessCNS}
For a minimizer $x^\star$ of~\eqref{eq:lassoA}, let $\Tx$ and $\ex$ be the subspace and vector in Definition~\ref{def:decomposable-norm} associated to $u^\star=L^* x^\star$, and denote $\Txper=\Tx^\perp$. $x^\star$ is the unique minimizer of~\eqref{eq:lassoA} if
\eq{
\dotp{L_{\Tx}^*h}{\ex} < \norma{L_{\Txper}^*h}^*, \qquad \forall h \in \ker(\Phi) \setminus \{0\} ~.
}
\end{thm}

The above condition is a strong generalization of the Null Space Property well known in $\ell_1$ regularization~\cite{DonohoHuo01}.

\subsection{Sufficient uniqueness conditions}
\subsubsection{General case}
A direct consequence of the above theorem is the following corollary.
\begin{cor}
\label{cor:uniqueness}
For a minimizer $x^\star$ of~\eqref{eq:lassoA}, let $\Tx$ and $\ex$ be the subspace and vector in Definition~\ref{def:decomposable-norm} associated to $u^\star=L^* x^\star$, and denote $\Txper=\Tx^\perp$. Assume that \SC{x^\star} is verified with $\norma{\alpha_{\Txper}}^* < 1$, and that \INJ{\Tx} holds. Then, $x^\star$ is the unique minimizer of~\eqref{eq:lassoA}.
\end{cor}
In fact, it turns out that the above two results are proved without requiring some restrictive implications of Definition~\ref{def:decomposable-norm}(ii) of decomposable norms, and are therefore valid for a much larger class of regularizations. This can be clearly checked in the arguments used in the proofs.

\subsubsection{Separable case}
\begin{defn}
\label{def:separable-norm}
The decomposable norm $\norma{\cdot}$ is separable on the subspace $T^\perp=S=V \oplus W \subset \RR^P$ if for any $u \in \RR^P$, $\norma{u_{T^\perp}}=\norma{u_V}+\norma{u_W}$.
\end{defn}
Separability as just defined is fulfilled for several decomposable norms such as the $\ell_1$ or $\ell_1-\ell_p$ norms, $1 \leq p < +\infty$.

The non-saturation condition on the dual certificate required in Corollary~\ref{cor:uniqueness} can be weakened to hold only on a subspace $V \subset \Txper$ and the conclusions of the corollary remain valid, and assuming a stronger restricted injectivity assumption. We have the following corollary. 
\begin{cor}
\label{cor:uniquenesssep}
Assume that $\norma{\cdot}$ is also separable, with $\Txper=V \oplus W$, such that \SC{x^\star} is verified with $\norma{\alpha_{V}}^* < 1$, and \INJ{V} holds. Then, $x^\star$ is the unique minimizer of~\eqref{eq:lassoA}.
\end{cor}

\section{Stability to noise}
\label{sec:stable}

\subsection{Main result}

\subsubsection{General case}
We are now ready to state our main stability results.
\begin{thm}
\label{theo:boundestimnonframe}
Let $\Txo$ and $\exo$ be the subspace and vector in Definition~\ref{def:decomposable-norm} associated to $u_0=L^* x_0$, and denote $\Txoper=\Txo^\perp$. Assume that \SC{x_0} is verified with $\norma{\alpha_{\Txoper}}^* < 1$, and that \INJ{\Txo} holds. Then, choosing $\lambda=c\varepsilon$, $c > 0$, the following holds for any minimizer $x^\star$ of~\eqref{eq:lassoA}
\eq{
\norm{x^\star-x_0}_2 \leq C\varepsilon ~,
}
where $C = C_1\pa{2+c\norm{\eta}_2}+C_2\tfrac{(1+c\norm{\eta}_2/2)^2}{c\pa{1-\norma{\alpha_{\Txoper}}^*}}$, and $C_{1} > 0$ and $C_2 > 0$ are constants independent of $\eta$ and $\alpha$.
\end{thm}

\begin{rem}[Separable case]
When the decomposable norm is also separable (see Corollary~\ref{cor:uniquenesssep}), the stability result of Theorem~\ref{theo:boundestimnonframe} remains true assuming that $\norma{\alpha_{V}}^*<1$ for $V \subset \Txoper$. This however comes at the price of the stronger restricted injectivity assumption \INJ{V}. To show this, the only thing to modify is the statement and the proof of Lemma~\ref{lem:nonsat} which can be done easily using similar arguments to those in the proof of Corollary~\ref{cor:uniquenesssep}.
\end{rem}

\subsubsection{Case of frames}
Suppose that $L^*$ is the analysis operator of a frame ($\ker(L^*)=\{0\}$) with lower bound $a > 0$, let $\tilde{L}$ be a dual frame. The following stability bound can be obtained whose proof is omitted for space limitations.

\begin{prop}\label{prop:boundestimframe}
Let $\Txo$ and $\exo$ be the subspace and vector in Definition~\ref{def:decomposable-norm} associated to $u_0=L^* x_0$, and denote $\Txoper=\Txo^\perp$. Assume that \SC{x_0} is verified with $\norma{\alpha_{\Txoper}}^* < 1$, and that $\Phi$ is injective on $\Im({\tilde{L}}_{\Txo})$. Then, choosing $\lambda=c\varepsilon$, $c > 0$, the following holds for any minimizer $x^\star$ of~\eqref{eq:lassoA}
\eq{
\norm{x^\star-x_0}_2 \leq C'\varepsilon ~,
}
where $C' = C_1\pa{2+c\norm{\eta}_2}+C_2'\tfrac{(1+c\norm{\eta}_2/2)^2}{c\pa{1-\norma{\alpha_{\Txoper}}^*}}$, and $C_{1} > 0$ and $C_2' > 0$ are constants independent of $\eta$ and $\alpha$.
\end{prop}

Since $\ker(L_{\Txoper}^*) \subseteq \Im(\tilde{L}_{\Txo})$, the required restricted injectivity assumption is more stringent than \INJ{\Txo}. On the positive side, the constant $C_2'$ is in general better than $C_2$. More precisely, the constant $C_{L}$, see the proof of Theorem~\ref{theo:boundestimnonframe}, is replaced with $\sqrt{a}$. Note also that coercivity of $R$ in this case allows to derive a bound similar to ours from the results in~\cite{GrasmairPositively11}. His restricted injectivity assumption is however different and our constants are sharper.

\subsection{Generalized irrepresentable condition}

In the following corollary, we provide a stronger sufficient stability condition that can be viewed as a generalization of the irrepresentable condition introduced in~\cite{fuchs-redundant-bases} when $R$ is the $\ell_1$ norm. It allows to construct dual vectors $\eta$ and $\alpha$ which obey the source condition and are computable, which in turn yield explicit constants in the bound. 

\begin{defn}
Let $T \subset \RR^P$ and $e \in \RR^P$, and denote $S=T^\perp$. Suppose that \INJ{T} is verified. Define for any $u \in \ker(L_{S})$ and $z \in \RR^{M}$ such that $\Phi^*z \in \Im(L_{S})$
\eq{
\IC_{u,z}(T,e) &=& \norma{\GT e + u_{S} + (L_{S})^+\Phi^*z}^*
}  
where
\eq{
\GT &=& (L_{S})^+(\Phi^* \Phi \AT - \Id)L_{T} \\
\AT : h \mapsto \AT h &=& \uargmin{x \in \ker(L_{S}^*)} \tfrac{1}{2}\norm{\Phi x}_2^2 - \dotp{h}{x} ~,
}
and $M^+$ is the Moore-Penrose pseudoinverse of $M$.
Let $\bar{u}$, $\bar{z}$ and $\underline{u}$ defined as
\eq{\label{eq:IC}
           &(\bar{u},\bar{z}) &= \uargmin{u \in \ker(L_{S}), \enscond{z}{\Phi^*z \in \Im(L_{S})}} \IC_{u,z}(T,e) \\
\text{and} &\underline{u} &= \uargmin{u \in \ker(L_{S})} \IC_{u,0}(T,e) ~.
}
\end{defn}
Obviously, we have
\eq{
\IC_{\bar{u},\bar{z}}(T,e) \leq \IC_{\underline{u},0}(T,e) \leq \IC_{0,0}(T,e) ~.
}
The convex programs defining $\IC_{\bar{u},\bar{z}}(T,e)$ and $\IC_{\underline{u},0}(T,e)$ can be solved using primal-dual proximal splitting algorithms whenever the proximity operator of $\norma{\cdot}$ can be easily computed~\cite{chambolle2011first}.
The criterion $\IC_{\underline{u},0}(T,e)$ specializes to the one developed in~\cite{VaiterAnalysis} when $\norma{\cdot}$ is the $\ell_1$ norm. $\IC_{0,0}(T,e)$ is a generalization of the coefficient involved in the irrepresentable condition introduced in~\cite{fuchs-redundant-bases} when $R$ is the $\ell_1$ norm, and to the one in \cite{CandesDecomposable12} for decomposable priors with $L=\Id$.


\begin{cor}
\label{cor:boundestimICnonframe}
Assume that \INJ{\Txo} is verified and $\IC_{\bar{u},\bar{z}}(\Txo,\exo) < 1$. Then, taking $\eta = \Phi\AT L_{\Txo} \exo+\bar{z}$, one can construct $\alpha$ such that \SC{x_0} is satisfied and $\norma{\alpha_{\Txoper}}^* < 1$. Moreover, the conclusion of Theorem~\ref{theo:boundestimnonframe} remains true substituting $1-\IC_{\bar{u},\bar{z}}(\Txo,\exo)$ for $1-\norma{\alpha_{\Txoper}}^*$.
\end{cor}

\section{Proofs}
\label{sec:proofs}

\subsection{Proof of Theorem~\ref{theo:uniquenessCNS}}
A key observation is that by strong (hence strict) convexity of $\mu \mapsto \norm{y - \mu}_2^2$, all minimizers of~\eqref{eq:lassoA} share the same image under $\Phi$. Therefore any minimizer of~\eqref{eq:lassoA} takes the form $x^\star+h$ where $h \in \ker(\Phi)$. Furthermore, it can be shown by arguments from convex analysis that any proper convex function $R$ has a unique minimizer $x^\star$ (if any) over a convex set $C$ if its directional derivative satisfies
\eq{
R'(x^\star;x-x^\star) > 0, \quad x \in C, x \neq x^\star ~.
}
Applying this to~\eqref{eq:lassoA} with $C=x^\star+\ker(\Phi)$, and using the fact that the directional derivative is the support function of the subdifferential, we get that $x^\star$ is the unique minimizer of~\eqref{eq:lassoA} if $\forall~ h \in \ker(\Phi) \setminus \{0\}$
\eq{
0 < R'(x^\star;h) &=& \sup_{v \in \partial R(x^\star)} \dotp{v}{h} \\
				 &=& \sup_{\alpha \in \partial \norma{\cdot}(L^* x^\star)} \dotp{\alpha}{L^*h} \\
				 &=& \dotp{\ex}{L_{\Tx}^*h} + \sup_{\norma{\alpha_{S}}^* \leq 1} \dotp{\alpha_{S}}{L_{\Txper}^*h} \\
				 &=& \dotp{\ex}{L_{\Tx}^*h} + \norma{L_{\Txper}^*h} ~.
}
We conclude using symmetry of the norm and the fact that $\ker(\Phi)$ is a subspace. 
\endIEEEproof

\subsection{Proof of Corollary~\ref{cor:uniqueness}}
The source condition \SC{x^\star} implies that $\forall~ h \in \ker(\Phi) \setminus \{0\}$
\eq{
\dotp{h}{L\alpha} = \dotp{h}{\Phi^*\eta} = \dotp{\Phi h}{\eta} = 0 ~.
}
Moreover 
\eq{
\begin{split}
\dotp{h}{L\alpha} = \dotp{L^* h}{\alpha} = & \dotp{L_{\Tx}^* h}{\ex} + \dotp{L_{\Txper}^* h}{\alpha_{S}} ~.
\end{split}
}
Thus, applying the dual-norm inequality we get
\eq{
\dotp{L_{\Tx}^* h}{\ex} \leq \norma{L_{\Txper}^* h}\norma{\alpha_{S}}^* < \norma{L_{\Txper}^* h} ~,
}
where the last inequality is strict since $L_{\Txper}^* h$ does not vanish owing to \INJ{\Tx}, and $\norma{\alpha_{S}}^* < 1$.
\endIEEEproof

\subsection{Proof of Corollary~\ref{cor:uniquenesssep}}
We follow the same lines as the proof of Corollary~\ref{cor:uniqueness} and get
\eq{
\begin{split}
\dotp{L^* h}{\alpha} = \dotp{L_{\Tx}^* h}{\ex} &+ \dotp{L_{V}^* h}{\alpha_{V}} + \dotp{L_{W}^* h}{\alpha_{W}}~.
\end{split}
}
We therefore obtain
\eq{
\begin{split}
\dotp{L_{\Tx}^* h}{\ex} &\leq \norma{L_{V}^* h}\norma{\alpha_{V}}^* + \norma{L_{W}^* h}\norma{\alpha_{W}}^* \\
							 &< \norma{L_{V}^* h} + \norma{L_{W}^* h} 
							 = \norma{L_{\Txper}^* h}^* ~,
\end{split}
}
where we used that $h \notin \ker(L_{V}^*)$, $\norma{\alpha_{V}}^* < 1$, separability and $\norma{\alpha_{W}}^* \leq \norma{\alpha_{V}}^* + \norma{\alpha_{W}}^*= \norma{\alpha_{S}}^*\leq 1$.
\endIEEEproof

\subsection{Proof of Theorem~\ref{theo:boundestimnonframe}}
We first define the Bregman distance/divergence.
\begin{defn}
Let $D^\reg_s(x,x_0)$ be the Bregman distance associated to $R$ with respect to $s \in \partial \reg(x_0)$, 
\[
D^\reg_s(x,x_0) = \reg(x) - \reg(x_0) - \dotp{s}{x-x_0} ~.
\]
Define $D^\Aa_\alpha(u,u_0)$ as the Bregman distance associated to $\norma{\cdot}$ with respect to $\alpha \in \partial \norma{\cdot}(u_0)$.
\end{defn}
Observe that by convexity, the Bregman distance is non-negative.

\myparagraph{Preparatory lemmata}
We first need the following key lemmata.
\begin{lem}[Prediction error and Bregman distance convergence rates]
\label{lem:boundprediction}
Suppose that \SC{x_0} is satisfied. Then, for any minimizer $x^\star$ of~\eqref{eq:lassoA}, and with $\la=c\varepsilon$ for $c > 0$, we have
\eq{
D^\reg_{\Phi^*\eta}(x^\star,x_0) = D^\Aa_\alpha(L^*x^\star,L^*x_0) &\leq& \varepsilon \frac{\pa{1+c\norm{\eta}_2/2}^2}{c} ~,\\
\norm{\Phi x^\star-\Phi x_0}_2 &\leq& \varepsilon(2+c\norm{\eta}_2) ~.
}
\end{lem}
The proof follows the same lines as that for any sublinear regularizer, see e.g.~\cite{ScherzerBook09}, where we additionally use the source condition \SC{x_0} and $D^\reg_{\Phi^*\eta}(x,x_0) = D^\reg_{L\alpha}(x,x_0) = D^\Aa_\alpha(L^*x,L^*x_0)$.\\

Now since $\norma{\cdot}$ is a norm, it is coercive, and thus
\[
\exists~ \cA > 0 \st \forall x \in \RR^P, ~ \norma{x} \geq \cA \norm{x}_2.
\]
We get the following inequality.
\begin{lem}[From Bregman to $\ell_2$ bound]
\label{lem:nonsat}
Suppose that \SC{x_0} holds with $\norma{\alpha_{\Txoper}}^* < 1$. Then,
\eq{
\norm{L_{\Txoper}^*(x^\star-x_0)}_2 \leq \frac{D^\Aa_\alpha(L^*x^\star,L^*x_0)}{\cA \pa{1-\norma{\alpha_{\Txoper}}^*}} ~,
}
\end{lem}

\begin{proof}
Decomposability of $\norma{\cdot}$ implies that $\exists v \in \Txoper$ such that $\norma{v}^* \leq 1$ and $\norma{L_{\Txoper}^*(x^\star-x_0)}=\dotp{L_{\Txoper}^*(x^\star-x_0)}{v}$. Moreover, $v+\exo \in \partial\norma{\cdot}(L^*x_0)$. Thus
\eq{
D^\Aa_\alpha(L^*x^\star,L^*x_0) &\geq& D^\Aa_\alpha(L^*x^\star,L^*x_0)  \\
							   && \qquad - D^\Aa_{v+\exo}(L^*x^\star,L^*x_0) \\
							   &=& \dotp{v+\exo-\alpha}{L^*(x^\star-x_0)} \\
							   &=& \dotp{v-\alpha_{\Txoper}}{L_{\Txoper}^*(x^\star-x_0)} \\ 
							   &=& \norma{L_{\Txoper}^*(x^\star-x_0)} \\
							   && \qquad - \dotp{\alpha_{\Txoper}}{L_{\Txoper}^*(x^\star-x_0)} \\
							   &\geq& \norma{L_{\Txoper}^*(x^\star-x_0)}(1-\norma{\alpha_{\Txoper}}^*) \\
							   &\geq& \cA\norm{L_{\Txoper}^*(x^\star-x_0)}_2(1-\norma{\alpha_{\Txoper}}^*) ~.
}
\end{proof}

\myparagraph{Proof of the main result}
\eq{
\norm{x^\star-x_0}_2 &\leq& \norm{\proj_{\ker(L_{\Txoper}^*)}(x^\star-x_0)}_2 \\
					&& \qquad + \norm{\proj_{\Im(L_{\Txoper}^*)}(x^\star-x_0)}_2 \\
				  &\leq& \icPhi\norm{\Phi\proj_{\ker(L_{\Txoper}^*)}(x^\star-x_0)}_2 \\
				  && \qquad +  \norm{\proj_{\Im(L_{\Txoper}^*)}(x^\star-x_0)}_2 \\
				  &\leq& \icPhi\norm{\Phi(x^\star-x_0)}_2 \\
				  && +  (1+\icPhi\norm{\Phi}_{2,2})\norm{\proj_{\Im(L_{\Txoper}^*)}(x^\star-x_0)}_2 ~,
}
where we used assumption \INJ{\Txo}, \ie 
\[
\exists~ \cPhi > 0 \st \norm{\Phi x}_2 \geq \cPhi \norm{x}_2, \quad \forall x \in \ker(L_{\Txoper}^*) ~.
\]

Since $L_{\Txoper}^*$ is injective on the orthogonal of its kernel, there exists $C_{L} > 0$ such that
\eq{
\norm{x^\star-x_0}_2 &\leq& \icPhi\norm{\Phi(x^\star-x_0)}_2 \\
				  && 	+  \tfrac{\norm{\Phi}_{2,2}+\cPhi}{C_{L}\cPhi}\norm{L_{\Txoper}^*\proj_{\Im(L_{\Txoper}^*)}(x^\star-x_0)}_2 ~.
}
Noticing that 
\[
\norm{L_{\Txoper}^*(x^\star-x_0)}_2=\norm{L_{\Txoper}^* \proj_{\Im(L_{\Txoper}^*)}(x^\star-x_0)}_2,
\] 
we apply Lemma~\ref{lem:nonsat} to get
\eq{
\norm{x^\star-x_0}_2 &\leq& \icPhi\norm{\Phi(x^\star-x_0)}_2  \\
				  && \quad +  \tfrac{\norm{\Phi}_{2,2}+\cPhi}{C_{L}\cPhi\pa{1-\norma{\alpha_{\Txoper}}^*}} D^\Aa_\alpha(L^*x^\star,L^*x_0)~.
}
Using Lemma~\ref{lem:boundprediction} yields the desired result.
\endIEEEproof

\subsection{Proof of Corollary~\ref{cor:boundestimICnonframe}}
Take
$\alpha = \exo + \GT \exo + \bar{u}_{\Txoper} + (L_{\Txoper})^+\Phi^*\bar{z}$.
First, $\alpha_{\Txo} = \exo$ since $\exo \in \Txo$ and $\Im(\GT) \subseteq \Im((L_{\Txoper})^+)=\Im(L_{\Txoper}^*)$.
Then $\norma{\alpha_{\Txoper}}^*=\IC_{\bar{u},\bar{z}}(\Txo,\exo) < 1$, whence we get that $\alpha \in \partial\norma{\cdot}(L^*x_0)$. 

Now, we observe by definition of $\AT$ that $\proj_{\ker(L_{\Txoper}^*)}(\Phi^* \Phi \AT - \Id)L_{\Txo} = 0$, which implies that $\Im((\Phi^* \Phi \AT - \Id)L_{\Txo})) \subseteq \Im(L_{\Txoper})$. In turn, $L_{\Txoper}\GT = L_{\Txoper}(L_{\Txoper})^+\pa{(\Phi^* \Phi \AT - \Id)L_{\Txo}} = \proj_{\Im(L_{\Txoper})}\pa{(\Phi^* \Phi \AT - \Id)L_{\Txo}} = (\Phi^* \Phi \AT - \Id)L_{\Txo}$. This, together with the fact that $\bar{u} \in \ker(L_{\Txoper})$ and $\Phi^*\bar{z} \in \Im(L_{\Txoper})$ yields
\eq{
L_{\Txoper} \alpha  
&=& (\Phi^* \Phi \AT - \Id)L_{\Txo} \exo + \Phi^* \bar{z} \\
&=& \Phi^*\eta - L_{\Txo} \alpha \iff \Phi^*\eta = L \alpha ~,
}
which implies that $\Phi^*\eta = L \alpha \in \partial\reg(x_0)$. We have just shown that the vectors $\alpha$ and $\eta$ as given above satisfy the source condition \SC{x_0} and the dual non-saturation condition. We conclude by applying Theorem~\ref{theo:boundestimnonframe} using \INJ{\Txo}.
\endIEEEproof

\section{Conclusion}
We provided a unified analysis of the structural properties of regularized solutions to linear inverse problems through a class of semi-norms formed by composing decomposable norms with a linear operator. We provided conditions that guarantee uniqueness, and also those ensuring stability to bounded noise. The stability bound was achieved without requiring (even partial) recovery of $\Txo$ and $\exo$. Recovery of $\Txo$ and $\exo$ for analysis-type decomposable priors and beyond is currently under investigation. Another perspective concerns whether the $\ell_2$ bound on $x^\star-x_0$ can be extended to cover more general low complexity-inducing regularizers beyond decomposable norms.


\bibliographystyle{plain}  
\bibliography{noise-robustness-decomposable}		

\end{document}